\documentclass[12pt,runningheads]{llncs}

\usepackage[paper=a4paper,left=2.8cm,right=2.8cm,top=2.8cm,bottom=2.8cm]{geometry}
\usepackage{amsmath, amssymb, mathtools}
\usepackage{relsize, hyperref, setspace, paralist, enumitem, abraces, xcolor}
\usepackage[english]{babel}
\newtheorem{protocol}{Protocol}

\newcommand{ \C }{ \mathcal C }
\newcommand{ \Z }{ \mathbb{Z} }
\newcommand{ \N }{ \mathbb{N} }
\newcommand{ \mya }{ \boldsymbol{a} }
\newcommand{ \myb }{ \boldsymbol{b} }
\newcommand{ \mye }{ \boldsymbol{e} }
\newcommand{ \myh }{ \boldsymbol{h} }
\newcommand{ \myH }{ \boldsymbol{H} }

\newcommand{ \myR }{ \boldsymbol{R} }

\newcommand{ \mys }{ \boldsymbol{s} }
\newcommand{ \myHH }{ \boldsymbol{\tilde{H}} }
\newcommand{ \myx }{ \boldsymbol{x} }
\newcommand{ \myy }{ \boldsymbol{y} }
\newcommand{ \bs }[1]{ \boldsymbol{#1} }
\newcommand{ \sml }{ \mathsmaller }
\newcommand{ \tnsp }{ {\sml{\sf T}} }

\newcommand{ \wtl }{ \operatorname{wt}_{\textsc{l}} }
\newcommand{ \distl }{ \operatorname{d}_{\textsc{l}} }
\newcommand{ \sign }{ \operatorname{sign} }

\newcommand{ \com }{ \operatorname{Com} }
\newcommand{ \defeq }{ \coloneqq }
\newcommand{ \myqed }{ \hfill $\blacktriangle$ }
\newcommand{\rvline}{\hspace*{-\arraycolsep}\vline\hspace*{-\arraycolsep}}

\begin{document}

\title{A Zero-Knowledge Proof for the Syndrome Decoding Problem in the {L}ee Metric}

%\titlerunning{A Zero-Knowledge Proof for the Lee Syndrome Decoding Problem}

\author{Mladen~Kova\v{c}evi\'{c}\thanks{Correspondence: \email{kmladen@uns.ac.rs}}\inst{1}
\and Tatjana~Grbi\'{c}\inst{1}
\and Darko~\v{C}apko\inst{1,2}
\and Nemanja~Nedi\'{c}\inst{1,2}
\and Sr\dj an~Vukmirovi\'{c}\inst{1,2}}

\authorrunning{Kova\v{c}evi\'{c} et al.}

\institute{Faculty of Technical Sciences, University of Novi Sad, Serbia \and
           Ethernal, Novi Sad, Serbia}

\maketitle

\begin{abstract}
The syndrome decoding problem is one of the NP-complete problems lying at the
foundation of code-based cryptography.
The variant thereof where the distance between vectors is measured with respect
to the Lee metric, rather than the more commonly used Hamming metric, has been
analyzed recently in several works due to its potential relevance for building
more efficient code-based cryptosystems.
The purpose of this article is to present a zero-knowledge proof of knowledge
for this variant of the problem.

\end{abstract}

\section{Introduction}
\label{sec:intro}

Code-based cryptography \cite{weger_book} is an area of study devoted to designing
and analyzing cryptographic primitives whose security rests on the assumptions of
intractability of certain computational problems arising in the theory of
error-correcting codes.
It originated from the works of McEliece \cite{mceliece} and Niederreiter \cite{niederreiter},
and is today one of the most promising and well-studied approaches to post-quantum
cryptography \cite{NIST}, so called because the corresponding hardness assumptions
are conjectured to hold even against adversaries having access to quantum computers.

In this paper, we focus on a computational problem called syndrome decoding of
random linear codes, which is one of the cornerstones of code-based cryptography.
In the literature, mainly the Hamming-metric version of the syndrome decoding
problem, whose NP-completeness was shown in \cite{berlekamp,barg}, has been studied.
In particular, zero-knowledge proofs and identification schemes for this and related
problems were pioneered by Stern \cite{stern}, and have been investigated extensively
later on, including some very recent works \cite{baldi,bidoux,feneuil}.
Lately, however, motivated by the possibility of reducing key sizes and building
more efficient code-based cryptosystems, other metrics have been explored as well
\cite{gaborit+zemor,puchinger}, in particular the Lee metric that we are currently
interested in \cite{chailloux,chailloux2,horlemann,horlemann2,li+wang,weger}.

The purpose of this article is to present a zero-knowledge proof of knowledge for
the Lee-metric version of the syndrome decoding problem.
The proof we shall describe falls within the permutation-based framework, a general
approach that has been previously applied to several problems, including the syndrome
decoding problem in the Hamming metric \cite{stern}, the knapsack and the subset-sum
problems \cite{shamir,blocki}, the short integer solution (SIS) problem \cite{ling}, etc.
The high-level idea behind it consists of revealing a permuted version of the secret,
while showing that the secret satisfies the desired linear relation using masking-based
combinatorial tricks.
While revealing a permuted version of the secret leaks no secret information in the case
of the Hamming-metric syndrome decoding problem or the subset-sum problem, it may be an
issue for other problems such as the Lee-metric version of syndrome decoding.
One of the main contributions of this article thus consists of rephrasing the Lee-metric
syndrome decoding problem into a form compatible with the permutation-based framework.

Zero-knowledge proofs are objects of great theoretical and practical interest
\cite{goldreich_book,thaler}.
Introduced in \cite{zkp}, they have found numerous applications ranging from
theoretical computer science to real-world security solutions.
The potential and usefulness of zero-knowledge protocols, especially those that
may be used as building blocks of quantum-resistant cryptographic primitives,
serve as the main motivation for the present work.
As an illustration of their usefulness, we note that several modern code-based
signature schemes are built on zero-knowledge protocols, e.g., CROSS \cite{baldi_NIST}
and LESS \cite{baldi_NIST2}, both included in the second round of the Additional
Digital Signature Schemes call in the NIST Post-Quantum Cryptography Standardization
Process.

\subsection*{Notation}

The ring of integers modulo $ m \geqslant 2 $, with standard modular operations
of addition and multiplication, is denoted by $ \Z_m \equiv \Z / m\Z $.
To simplify the exposition, we shall assume throughout the paper that $ m $ is an
odd number larger than $ 3 $:%
\footnote{For $ m = 2 $ and $ m = 3 $, the Lee metric on $ \Z_m^n $ is identical
to the Hamming metric, and zero-knowledge proofs for the syndrome decoding problem
in the Hamming metric are well-known \cite{stern}.}
\begin{align}
\label{eq:ml}
  m = 2\ell + 1 , \quad  \ell \geqslant 2 .
\end{align}
We emphasize, however, that the results are valid for even $ m $ as well.
The elements of $ \Z_m $ are assumed to be represented by the integers in the set
\begin{align}
\label{eq:mod}
  \left\{-\ell, \ldots, -1, 0, 1, \ldots, \ell \right\} .
\end{align}
This representation will be more convenient than the more common $ \{0, 1, \ldots, 2\ell\} $
for our purpose.

Vectors and matrices are written in boldface.
All vectors are thought of as row vectors.
The vectors consisting of all zeros (resp. ones) are denoted by $ \bs{0} $ (resp.
$ \bs{1} $), or $ \bs{0}_n $ (resp. $ \bs{1}_n $) if their length needs to be emphasized.
The $ n \times k $ matrix of all zeros is denoted by $ \bs{0}_{n\times k} $, and
the $ n \times n $ identity matrix by $ \bs{I}_n $.
For two vectors $ \bs{x} = (x_1, \ldots, x_n) \in \Z_m^n $,
$ \bs{y} = (y_1, \ldots, y_k) \in \Z_m^k $, their concatenation is a vector in
$ \Z_m^{n+k} $ denoted by $ (\bs{x} \, \vert\, \bs{y}) = (x_1, \ldots, x_n, y_1, \ldots, y_k) $.

The symbol $ \sum $ denotes the real sum.
To avoid confusion, we shall write $ (\operatorname{mod} m) $ whenever we want
to emphasize that operations are performed modulo $ m $.

\section{Codes in the Lee metric and the syndrome decoding problem}
\label{sec:lee}

The Lee weight \cite{lee,ulrich,chiang+wolf} of a vector
$ \myx = (x_1, \ldots, x_n) \in \Z_{m}^n $, under our convention \eqref{eq:mod}
that the elements $ x_i $ are represented by the integers in the range
$ -\ell, \ldots, \ell $, is defined as
\begin{equation}
\label{eq:wtl}
  \wtl(\myx) \defeq \sum_{i=1}^n |x_i| .
\end{equation}
In other words, the Lee weight of a vector is a modular version of its $ L_1 $
(or Manhattan) norm.
The maximum possible Lee weight of a vector in $ \Z_m^n $ is
\begin{align}
\label{eq:maxwtl}
  \max_{\bs{x} \in \Z_m^n}  \wtl(\bs{x})  =  n \ell .
\end{align}
The Lee distance between $ \myx, \myy \in \Z_m^n $ is defined as
\begin{equation}
\label{eq:dl}
  \distl(\myx, \myy) \defeq \wtl(\myx - \myy) .
\end{equation}

A code of length $ n $ over $ \Z_m $ is a nonempty subset of $ \Z_m^n $ of
cardinality at least two.
We say that $ \C \subseteq \Z_m^n $ is a linear code over $ \Z_m $ if it is a
$ \Z_m $-submodule of $ \Z_m^n $.
A matrix $ \bs{G} $ is said to be a generator matrix of a linear code $ \C $ if its
row-space equals $ \C $.
A matrix $ \bs{P} $ is said to be a parity-check matrix of a linear code $ \C $ if
its kernel equals $ \C $, i.e., if
\begin{align}
\label{eq:code}
  \C = \big\{ \bs{x} \in \Z_m^n \colon \bs{x} \bs{P}^{\tnsp} = \bs{0} \pmod{m} \big\} .
\end{align}
Generator and parity-check matrices of linear codes always exist, but are not unique.

We next state (the decision version) of the computational problem known as syndrome
decoding.
This problem has been studied in the literature mostly for the Hamming-metric case,
but it can naturally be considered in other settings as well.
We state it for the case we are currently interested in -- codes in the module
$ \Z_m^n $ endowed with the Lee metric.
This version of the problem has attracted some attention recently and is known to
be NP-complete \cite{weger}.

\begin{problem}[Lee syndrome decoding]
\label{problem:syn}
Given $ \myH \in \Z_m^{n \times (n-k)} $, $ \mys \in \Z_m^{n-k} $, and $ w \in \N $,
determine whether there is a vector $ \mye \in \Z_m^n $ satisfying the following
two conditions:
\begin{itemize}[leftmargin=2em,noitemsep,topsep=0pt]
\item[1.)]
$ \mye \myH = \mys \pmod{m} $,
\item[2.)]
$ \wtl(\mye) \leqslant w $.\myqed
\end{itemize}
\end{problem}

Problem~\ref{problem:syn} is motivated by the communication scenario in which the task
of the decoder is to find the codeword $ \bs{x} \in \C $ that is closest in Lee distance
to the received vector $ \bs{y} = \bs{x} + \bs{e} \pmod{m} $, the assumption being that
this is the most likely codeword to have been transmitted.
Since $ \distl(\bs{x}, \bs{y}) = \wtl(\bs{e}) $ (see~\eqref{eq:dl}), this is equivalent
to finding the error vector $ \bs{e} $ of smallest Lee weight such that $ \bs{y} - \bs{e} \in \C $.
To this end the decoder can compute the \emph{syndrome}
$ \bs{s} \defeq \bs{y} \bs{P}^{\tnsp} = \bs{e} \bs{P}^{\tnsp} \pmod{m} $ (see~\eqref{eq:code}),
where $ \bs{P} $ is the parity-check matrix of the code $ \C $, and search for the
error vector $ \bs{e} $ of smallest weight satisfying $ \bs{e} \bs{P}^{\tnsp} = \bs{s} \pmod{m} $.
The decision version of the problem just described is precisely Problem~\ref{problem:syn}.
(In the statement of Problem~\ref{problem:syn} and throughout the paper, we use the
notation $ \myH $ instead of $ \bs{P}^{\tnsp} $ for simplicity.)

\section{The balanced Lee syndrome decoding problem}

In this section, we introduce a minor modification of Problem~\ref{problem:syn}, which
will be needed for technical reasons but does not incur a loss in generality.

We say that a vector $ \mye = (e_1, \ldots, e_n) \in \Z_m^n $ is \emph{balanced} if the
weight of its positive entries equals the weight of its negative entries (recall
\eqref{eq:mod}), namely
\begin{align}
  \sum_{i \colon e_i > 0} |e_i|  =  \sum_{i \colon e_i < 0} |e_i| ,
\end{align}
which is equivalent to saying that $ \sum_{i=1}^n e_i = 0 $.
The Lee weight of a balanced vector is necessarily even, i.e., $ \wtl(\mye) \in 2 \N $.

\begin{problem}[Balanced Lee syndrome decoding]
\label{problem:balsyn}
Given $ \myH \in \Z_m^{n \times (n-k)} $, $ \mys \in \Z_m^{n-k} $, and $ w \in 2\N $
with $ w \leqslant n(\ell-1) $, determine whether there is a vector $ \mye \in \Z_m^n $
satisfying the following conditions:
\begin{itemize}[leftmargin=2em,noitemsep,topsep=0pt]
\item[1.)]
$ \mye \myH = \mys  \pmod{m} $,
\item[2.)]
$ \wtl(\mye) \leqslant w $,
\item[3.)]
$ \sum_{i=1}^n e_i = 0 $.\myqed
\end{itemize}
\end{problem}

Notice that we have also incorporated the condition $ w \leqslant n(\ell-1) $ in the
statement of the problem.
This condition will also be needed for technical reasons.

Neither the assumption that $ w \leqslant n(\ell-1) $, nor the requirement that the
vector $ \bs{e} $ is balanced, significantly affect the computational difficulty of
the syndrome decoding problem.
This fact is formalized in the following proposition.

\begin{proposition}
\label{thm:NPcomplete}
The balanced Lee syndrome decoding problem is NP-complete.
\end{proposition}
\begin{proof}
Let $ \myH \in \Z_m^{n\times (n-k)} $, $ \mys \in \Z_m^{n-k} $, $ w \in \N $ be an
instance of Problem~\ref{problem:syn}.
Let
\begin{align}
  \bar{n} = n + \Big\lceil \frac{n}{\ell-1} \Big\rceil
\end{align}
and note that $ \bar{n} \geqslant \frac{\ell}{\ell-1} n $.
Define $ \bs{\bar{H}} \in \Z_m^{\bar{n}\times (\bar{n}-k)} $ and $ \bs{\bar{s}} \in \Z_m^{\bar{n}-k} $
by:
\begin{align}
  \bs{\bar{H}} =
  \begin{pmatrix}
    \myH
    & \rvline & \, \bs{0}_{\sml{n \times \lceil\!\frac{n}{\ell-1}\!\rceil}} \\
  \hline
    \bs{0}_{\sml{\lceil\!\frac{n}{\ell-1}\!\rceil \times (n-k)}} \, & \rvline &
    \bs{I}_{\sml{\lceil\!\frac{n}{\ell-1}\!\rceil}}
  \end{pmatrix} , \qquad
	\bs{\bar{s}} = \big(\bs{s} \, \vert\, \bs{0}_{\sml{\lceil\!\frac{n}{\ell-1}\!\rceil}}\big) .
\end{align}
Then a vector $ \bs{\bar{e}} \in \Z_m^{\bar{n}} $ is a solution to
$ \bs{\bar{e}} \bs{\bar{H}} = \bs{\bar{s}} \pmod{m} $ if and only if it is of the form
$ \bs{\bar{e}} = \big(\bs{e} \,|\, \bs{0}_{\sml{\lceil\!\frac{n}{\ell-1}\!\rceil}}\big) $
and $ \bs{e} \bs{H} = \bs{s} \pmod{m} $.
Moreover,
\begin{align}
\label{eq:wtcond}
  \wtl(\bs{\bar{e}}) = \wtl(\bs{e}) \leqslant n \ell \leqslant \bar{n} (\ell - 1) .
\end{align}

Now define the matrix
\begin{align}
\label{eq:Hpm}
  \bs{H^{\sml{\pm}}} =
  \begin{pmatrix}
     \bs{\bar{H}}  \\
  \hline
     -\bs{\bar{H}}\,
  \end{pmatrix} .
\end{align}
The following argument will establish that $ (\myH, \mys, w) $ is a `yes' instance
of the Lee syndrome decoding problem if and only if $ (\bs{H^{\sml{\pm}}}, 2\bs{\bar{s}}, 2w) $
is a `yes' instance of the balanced Lee syndrome decoding problem, which will prove
the proposition.
To show the `only if' direction, simply notice that
\begin{align}
  \bs{e} \bs{H} = \bs{s}
	\quad \Longrightarrow \quad
	\bs{\bar{e}} \bs{\bar{H}} = \bs{\bar{s}}
  \quad \Longrightarrow \quad
  (\bs{\bar{e}} \, \vert -\!\bs{\bar{e}}) \bs{H^{\sml{\pm}}} = 2 \bs{\bar{s}} ,
\end{align}
that the vector $ (\bs{\bar{e}} \, \vert -\!\bs{\bar{e}}) $ is balanced, that its
length is $ 2 \bar{n} $, and that, by \eqref{eq:wtcond}, its weight satisfies
\begin{align}
  \wtl((\bs{\bar{e}} \, \vert -\!\bs{\bar{e}})) = 2 \wtl(\bs{\bar{e}}) \leqslant 2 \bar{n} (\ell - 1) .
\end{align}
Therefore, if $ \bs{e} $ is a solution to the instance $ (\myH, \mys, w) $ of
Problem~\ref{problem:syn}, then $ (\bs{\bar{e}} \, \vert -\!\bs{\bar{e}}) $ is a
solution to the instance $ (\bs{H^{\sml{\pm}}}, 2\bs{\bar{s}}, 2w) $ of Problem~\ref{problem:balsyn}.
For the `if' direction, suppose that $ \bs{g} \in \Z_m^{2\bar{n}} $ satisfies
$ \bs{g} \bs{H^{\sml{\pm}}} = 2 \bs{\bar{s}} \pmod{m} $, and denote by
$ \bs{g}_{\textnormal{l}}, \bs{g}_{\textnormal{r}} \in \Z_m^{\bar{n}} $
the left and right half of $ \bs{g} $, i.e.,
$ \bs{g} = (\bs{g}_{\textnormal{l}} \,\vert\, \bs{g}_{\textnormal{r}}) $.
Then, taking into account the form of $ \bs{H^{\sml{\pm}}} $ (see \eqref{eq:Hpm})
and dividing by $ 2 $ (which is possible because $ \gcd(m,2) = 1 $ by the assumption
\eqref{eq:ml}), we get
$ 2^{-1}(\bs{g}_{\textnormal{l}} - \bs{g}_{\textnormal{r}}) \bs{\bar{H}} = \bs{\bar{s}} \pmod{m} $.
As argued in the first paragraph of the proof, the vector
$ 2^{-1}(\bs{g}_{\textnormal{l}} - \bs{g}_{\textnormal{r}}) \in \Z_m^{\bar{n}} $ is
then necessarily of the form
$ \big(\bs{e} \,|\, \bs{0}_{\sml{\lceil\!\frac{n}{\ell-1}\!\rceil}}\big) $, with
$ \bs{e} \in \Z_m^n $ satisfying $ \bs{e} \bs{H} = \bs{s} \pmod{m} $.
%Uzeli smo ovakvu H', a ne npr. $ \frac{H}{H} $ da ne bi pp da je $ k < n/2 $.
\hfill $ \blacksquare $
\end{proof}

The protocol we shall describe in the following section requires restating the
balanced Lee syndrome decoding problem in an appropriate way.
Given $ \bs{H} \in \Z_m^{n \times (n-k)} $, let $ \myHH \in \Z_m^{n\ell \times (n-k)} $
be the matrix obtained by repeating each of the row vectors of $ \myH $ exactly
$ \ell = \lfloor m/2 \rfloor $ times.
Namely, if $ \myh_{1}, \ldots, \myh_{n} $ are the row vectors of $ \myH $, we set%
\begin{align}
\label{eq:Htilde}
  \myH =
	\begin{pmatrix}
  \myh_{1}  \\
	\vdots  \\
	\myh_{n}  \\
	\end{pmatrix}
	\quad  \longrightarrow  \quad
  \myHH =
	\begin{pmatrix}
	%\begin{rcases}
	%\begin{matrix}
	%\myh_1  \\
	%\vdots  \\
	%\myh_1
	%\end{matrix}
	%\end{rcases}
	%\text{$ \lceil \frac{m-1}{2} \rceil $}  \\
  \myh_{1}  \\
	\vdots  \\
	\myh_{1}  \\
	\vdots  \\
	\myh_{n}  \\
	\vdots  \\
	\myh_{n}
	\end{pmatrix}.
\end{align}
Proposition~\ref{thm:equiv} below establishes that Problem \ref{problem:balsyn} can
equivalently be stated as follows.

\begin{problem}[Balanced Lee syndrome decoding -- alternative formulation]
\label{problem:alt}
Given $ \myH \in \Z_m^{n \times (n-k)} $, $ \mys \in \Z_m^{n-k} $, and $ w \in 2\N $
with $ w \leqslant n(\ell-1) $, determine whether there is a vector $ \bs{f} \in \{-1, 0, 1\}^{n\ell} $
satisfying the following conditions:
\begin{itemize}[leftmargin=2em,noitemsep,topsep=0pt]
\item[1.)]
$ \bs{f} \myHH = \mys  \pmod{m} $, where $ \myHH $ is the matrix defined in \eqref{eq:Htilde},
\item[2.)]
$ \sum_{j=1}^{n\ell} |f_j| = w $,
\item[3.)]
$ \sum_{j=1}^{n\ell} f_j = 0 $.\myqed
\end{itemize}
\end{problem}

\begin{proposition}
\label{thm:equiv}
Fix $ \myH \in \Z_m^{n\times (n-k)} $, $ \mys \in \Z_m^{n-k} $, and $ w \in 2\N $ with
$ w \leqslant n(\ell-1) $.
Let $ \myHH $ be the matrix defined in \eqref{eq:Htilde}.
Then there exists a vector $ \mye \in \Z_m^n $ satisfying $ \mye \myH = \mys \pmod{m} $,
$ \wtl(\mye) \leqslant w $, and $ \sum_{i=1}^n e_i = 0 $, if and only if there exists
a vector $ \bs{f} \in \{-1, 0, 1\}^{n\ell} $ satisfying $ \bs{f} \myHH = \mys \pmod{m} $,
$ \wtl(\bs{f}) = w $, and $ \sum_{j=1}^{n\ell} f_j = 0 $.
\end{proposition}
\begin{proof}
We first prove the `only if' direction.
Let $ \mye $ be the vector satisfying $ \mye \myH = \mys \pmod{m} $,
$ \wtl(\mye) \leqslant w $, and $ \sum_{i=1}^n e_i = 0 $.
Define $ \bs{e'} \in \{-1, 0, 1\}^{n \ell} $ as the vector obtained by replacing
each coordinate of $ \mye = (e_1, \ldots, e_n) \in \Z_m^n $ with a block from
$ \{-1, 0, 1\}^{\ell} $ consisting of $ |e_i| $ $ 1 $'s (resp.\ $ -1 $'s) if
$ e_i > 0 $ (resp.\ $ e_i < 0 $), followed by $ \ell - |e_i| $ zeros.
In other words, recalling the definition of the sign function
\begin{equation}
  \sign(x) \defeq
	\begin{cases}
	  +1 , &\textnormal{if}\ x > 0  \\
		\phantom{+}0  , &\textnormal{if}\ x = 0  \\
		-1 , &\textnormal{if}\ x < 0
	\end{cases}
\end{equation}
we set, for each $ i = 1, \ldots, n $,
\begin{subequations}
\begin{alignat}{3}
  e'_{(i-1) \ell + 1}          &= \ldots  &&= e'_{(i-1) \ell + |e_i|}  &&= \sign(e_i)  \\
	e'_{(i-1) \ell + |e_i| + 1}  &= \ldots  &&= e'_{i \ell}              &&=  0 .
\end{alignat}
\end{subequations}
Henceforth, the coordinates $ e'_{(i-1) \ell + 1}, \ldots, e'_{i \ell} $ will together
be called the $ i $'th \emph{block} of $ \bs{e'} $.
Note that
\begin{align}
  \sum_{i=1}^n e_i = 0  \quad \Longleftrightarrow \quad  \sum_{j=1}^{n \ell} e'_j = 0 ,
\end{align}
so that $ \bs{e'} $ is also balanced,
\begin{align}
\label{eq:weightf}
  \wtl(\mye) = \sum_{i=1}^n |e_i| = \sum_{j=1}^{n\ell} |e'_j| = \wtl(\bs{e'})  ,
\end{align}
and
\begin{align}
\label{eq:eg}
  \mye  \myH = \bs{e'} \myHH .
\end{align}
The latter is due to the fact that
\begin{equation}
\label{eq:sumhi}
  e_i \myh_{i} =
	\begin{cases}
	  \underbrace{+\myh_{i} + \ldots + \myh_{i}}_{|e_i|} ,  &\textnormal{if}\ e_i > 0  \\
		\underbrace{-\myh_{i} - \ldots - \myh_{i}}_{|e_i|} ,  &\textnormal{if}\ e_i < 0
  \end{cases}
\end{equation}
and that $ |e_i| \leqslant \ell $ (see \eqref{eq:mod}).
Hence, $ \bs{e'} $ satisfies the claimed properties except that its Lee weight
is $ \leqslant\! w $, but not necessarily $ =\! w $ (see \eqref{eq:weightf}).
We therefore modify $ \bs{e'} $ to obtain a vector $ \bs{e''} $ that satisfies the
weight property as well.
If $ \wtl(\bs{e'}) = w $, we set $ \bs{e''} = \bs{e'} $.
Otherwise, if $ \wtl(\bs{e'}) < w $, we find a block in $ \bs{e'} $ that contains
at least two zeros and replace the leftmost two zeros in that block with the pair
$ +1, -1 $.
We then repeat this procedure until a vector of Lee weight $ w $ is obtained.
Note that it is possible to do this because blocks are of size $ \ell \geqslant 2 $
(see \eqref{eq:ml}) and because a pair of zeros belonging to the same block must
exist in a vector of weight $ <\! n(\ell-1) $ (this is the reason we imposed the
condition $ w \leqslant n(\ell-1) $).
Clearly, $ \bs{e''} $ is also balanced and
\begin{align}
  \wtl(\bs{e''}) = \sum_{j=1}^{n\ell} |e''_j| = w .
\end{align}
In other words, $ \bs{e''} $ has exactly $ w/2 $ ones and $ w/2 $ negative ones.
Furthermore, we have
\begin{align}
\label{eq:gf}
	\bs{e'} \myHH = \bs{e''} \myHH ,
\end{align}
which follows from the fact that the product $ \bs{e''} \myHH $, compared to $ \bs{e'} \myHH $,
just contains additional terms of the form $ \pm \myh_{i} $.
Therefore, we may take $ \bs{e''} $ as the vector $ \bs{f} $ from the statement of the
proposition.

To prove the `if' direction, let $ \bs{f} \in \{-1, 0 , 1\}^{n\ell} $ be any vector
satisfying $ \bs{f} \myHH = \mys \pmod{m} $, $ \sum_{j=1}^{n\ell} |f_j| = w $, and
$ \sum_{j=1}^{n\ell} f_j = 0 $.
Define the vector $ \mye = (e_1, \ldots, e_n) $ as the block-wise accumulation of $ \bs{f} $,
namely
\begin{align}
  e_i = \sum_{j=(i-1) \ell + 1}^{i\ell} f_j , \qquad i = 1, \ldots, n ,
\end{align}
and note that $ e_i \in \{-\ell, \ldots, \ell\} $.
Then
\begin{align}
  \bs{e} \myH = \bs{f} \myHH = \mys \pmod{m} ,
\end{align}
\begin{align}
	\sum_{i=1}^n e_i = \sum_{i=1}^n \sum_{j=(i-1) \ell + 1}^{i\ell} f_j = \sum_{j=1}^{n\ell} f_j = 0 ,
\end{align}
and
\begin{align}
  \wtl(\mye) = \sum_{i=1}^n |e_i| = \sum_{i=1}^n \bigg| \sum_{j=(i-1) \ell + 1}^{i\ell} f_j \bigg|
	\leqslant  \sum_{j=1}^{n\ell} |f_j| = w ,
\end{align}
as claimed.
\hfill $ \blacksquare $
\end{proof}

For the purpose of illustration, we provide a concrete example of the vectors
$ \bs{e'} $ and $ \bs{e''} $ from the above proof.
The vector $ \bs{e''} $ will be crucial for the formulation of our zero-knowledge
protocol.

\begin{example}
\label{example}
Let $ n = 6 $, $ m = 7 $, and $ \mye = (-2, 0, 1, 3, -1, -1) \in \Z_7^6 $.
Then $ \ell = 3 $ and
\begin{align}
  \bs{e'} = (-1, -1, 0 \,|\, 0, 0, 0 \,|\, 1, 0, 0 \,|\, 1, 1, 1 \,| -\!1, 0, 0 \,| -\!1, 0, 0) \in \{-1, 0 , 1\}^{18} ,
\end{align}
where we have separated the blocks of $ \bs{e'} $ by vertical lines for clarity.
Suppose that, in the given instance of the syndrome decoding problem, $ w = 10 $.
Then $ \wtl(\bs{e'}) = \wtl(\bs{e}) = 8 < w $, so we need to modify $ \bs{e'} $
by the procedure described in the proof of Proposition \ref{thm:equiv} in order
to increase its Lee weight to $ w = 10 $.
We obtain
\begin{align}
  \bs{e''} = (-1, -1, 0 \,|\, {\color{blue}1, -1}, 0 \,|\, 1, 0, 0 \,|\, 1, 1, 1 \,| -\!1, 0, 0 \,| -\!1, 0, 0) \in \{-1, 0 , 1\}^{18}
\end{align}
with $ \wtl(\bs{e''}) = 10 = w $.
\myqed
\end{example}

\section{Zero-knowledge proof for the Lee syndrome decoding problem}
\label{sec:zkp}

We first recall the definitions of zero-knowledge proofs and commitment schemes.
For reasons of clarity and simplicity, we chose not to overburden the exposition
with too many technical details.
For example, we do not discuss and distinguish between the notions of computational/%
statistical/perfect security.
See, e.g., \cite{goldreich_book,thaler} for a fully formal treatment of these notions.

An interactive proof for a language $ L $ is a protocol whereby one party, the
Prover, can convince another party, the Verifier, that a given instance $ x $
belongs to $ L $.
Multiple rounds of interaction between the two parties are allowed, the Verifier
is assumed to be a probabilistic polynomial-time algorithm, and the following
properties are required to hold:
\begin{itemize}[noitemsep,topsep=0pt]
\item
\emph{Completeness}: If $ x \in L $, the Verifier will output `accept' with
probability $ 1 $.
\item
\emph{Soundness}: If $ x \notin L $, the Verifier will output `reject' with
probability at least $ 1 - \epsilon $, regardless of the strategy and the
computational resources of the Prover.
The parameter $ \epsilon $ is then called the soundness error.
\end{itemize}
An interactive proof is said to be a zero-knowledge (ZK) proof if it, in addition
to the above, possesses the following property:
\begin{itemize}[noitemsep,topsep=0pt]
\item
\emph{Zero-knowledge}: During the execution of the protocol, the Verifier learns
nothing beyond the fact that $ x \in L $.
\end{itemize}
Put differently, an interactive proof is zero-knowledge if the Verifier learns
nothing from the interaction with the Prover that it could not have computed by
itself prior to the interaction.

A commitment scheme \cite{blum} is one of the basic cryptographic primitives
commonly used in the design of zero-knowledge proofs \cite{goldreich}.
It enables one party, say the above-mentioned Prover, to choose an element $ z $
of a predefined set and send a ``masked'' version of it, $ \com(z) $, to the Verifier,
so that the following properties hold:
\begin{itemize}[leftmargin=2em,noitemsep,topsep=0pt]
\item
\emph{Binding property}: The Prover cannot change its mind at a later point, i.e.,
after the Prover decides to ``open'' $ \com(z) $, the Verifier will be confident
(with probability $ \approx\! 1 $) that the value it now sees is indeed the one the
Prover originally committed to.
\item
\emph{Hiding property}: The Verifier cannot infer anything about $ z $ from $ \com(z) $
alone (except with negligible probability).
\end{itemize}

\subsection{Zero-knowledge protocol for the balanced Lee syndrome decoding problem}

The protocol described next is inspired by the zero-knowledge proof for the
subset sum problem \cite{blocki}.
The idea behind it is based on the following general principle:
by using randomization, the Prover creates several possibilities which, together,
uniquely determine the secret solution (a.k.a.\ witness), but either of which
individually is random and gives no information about the solution.
Then the Prover can convince the Verifier that it possesses the solution by letting
the Verifier choose which of the possibilities it wishes the Prover to reveal.

\begin{protocol}[ZK proof for Problem \ref{problem:balsyn}]
\label{protocol}
\textnormal{
Let $ \myH \in \Z_m^{n\times (n-k)} $, $ \mys \in \Z_m^{n-k} $, and $ w \in 2\N $ with
$ w \leqslant n(\ell-1) $, be an instance of Problem~\ref{problem:balsyn} (given to
both the Prover and the Verifier), and $ \mye \in \Z_m^{n} $ a solution/witness (known
by the Prover only) satisfying $ \mye \myH = \mys  \pmod{m} $, $ \wtl(\mye) \leqslant w $,
and $ \sum_{i=1}^n e_i = 0 $.
Let $ \com(\cdot) $ be a commitment scheme.
\begin{enumerate}[itemsep=2mm]
\item
The Prover creates a vector $ \bs{f} \in \{-1, 0, 1\}^{n \ell} $ from the vector
$ \bs{e} $ as described in the proof of Proposition~\ref{thm:equiv}, by setting
$ \bs{f} = \bs{e''} $.
\item
The Prover selects a matrix $ \myR \in \Z_m^{n\times(n-k)} $ uniformly at random,
computes $ \bs{T} = \myH - \myR \pmod{m} $, and creates the matrices
$ \bs{\tilde{R}}, \bs{\tilde{T}} \in \Z_m^{n\ell\times(n-k)} $ from $ \bs{R}, \bs{T} $
as in \eqref{eq:Htilde}.
\item
The Prover computes $ \mya \defeq \mye \myR = \bs{f} \bs{\tilde{R}} \pmod{m} $ and
$ \myb \defeq \mye \bs{T} = \bs{f} \bs{\tilde{T}} \pmod{m} $ (see \eqref{eq:eg} and
\eqref{eq:gf}).
\item
The Prover selects a permutation $ \pi $ on $ \{1, \ldots, n\ell\} $ uniformly at random.
Let $ \bs{\tilde{R}}_\pi, \bs{\tilde{T}}_{\!\pi} $ be the row-permuted versions of
the matrices $ \bs{\tilde{R}}, \bs{\tilde{T}} $,
\begin{align}
  \bs{\tilde{R}} =
	\begin{pmatrix}
  \bs{\tilde{r}}_{1}  \\
	\vdots  \\
	\bs{\tilde{r}}_{n\ell}  \\
	\end{pmatrix}
	\quad \longrightarrow \quad
  \bs{\tilde{R}}_{\pi} =
	\begin{pmatrix}
  \bs{\tilde{r}}_{\pi(1)}  \\
	\vdots  \\
	\bs{\tilde{r}}_{\pi(n\ell)}  \\
	\end{pmatrix} ,
\end{align}
and $ \bs{f}_{\!\pi} = (f_{\pi(1)}, \ldots, f_{\pi(n\ell)}) $ the permuted version
of the vector $ \bs{f} $.
\item
The Prover computes the commitments
$ \com(\bs{R}) $, $ \com(\bs{T}) $, $ \com(\mya) $, $ \com(\myb) $, $ \com(\pi) $,
$ \com\!\big(\bs{\tilde{R}}_\pi\big) $, $ \com\!\big(\bs{\tilde{T}}_{\!\pi}\big) $,
$ \com(\bs{f}_{\!\pi}) $, and sends them to the Verifier.
\item
The Verifier can ask the Prover to do exactly one of the following:
\begin{enumerate}[topsep=2mm]
\item
Open the commitments $ \com(\bs{R}) $, $ \com(\bs{T}) $, $ \com(\pi) $,
$ \com\!\big(\bs{\tilde{R}}_\pi\big) $, $ \com\!\big(\bs{\tilde{T}}_{\!\pi}\big) $.
After receiving the reply, the Verifier checks whether the following conditions
are met:
\begin{enumerate}[topsep=2mm,leftmargin=1cm]
\item[(a1)]
$ \bs{R} + \bs{T} \stackrel{?}{=} \bs{H} \pmod{m} $,
\item[(a2)]
$ \bs{\tilde{R}}_\pi $, $ \bs{\tilde{T}}_{\!\pi} $ are formed properly,
i.e., by forming $ \bs{\tilde{R}}, \bs{\tilde{T}} $ from $ \bs{R}, \bs{T} $ as in
\eqref{eq:Htilde}, and then permuting their rows according to $ \pi $.
\end{enumerate}
\item
Open the commitments $ \com(\mya) $, $ \com(\myb) $,
$ \com\!\big(\bs{\tilde{R}}_\pi\big) $, $ \com(\bs{f}_{\!\pi}) $.
After receiving the reply, the Verifier checks whether the following conditions
are met:
\begin{enumerate}[topsep=2mm,leftmargin=1cm]
\item[(b1)]
$ \bs{a} + \bs{b} \stackrel{?}{=} \bs{s} \pmod{m} $,
\item[(b2)]
$ \bs{f}_{\!\pi} \bs{\tilde{R}}_\pi \stackrel{?}{=} \bs{a} \pmod{m} $,
\item[(b3)]
$ \wtl(\bs{f}_{\!\pi}) \stackrel{?}{=} w $,
\item[(b4)]
$ \bs{f}_{\!\pi} $ is balanced.
\end{enumerate}
\item
Open the commitments $ \com(\mya) $, $ \com(\myb) $,
$ \com\!\big(\bs{\tilde{T}}_{\!\pi}\big) $, $ \com(\bs{f}_{\!\pi}) $.
After receiving the reply, the Verifier checks whether the following conditions
are met:
\begin{enumerate}[topsep=2mm,leftmargin=1cm]
\item[(c1)]
$ \bs{a} + \bs{b} \stackrel{?}{=} \bs{s} \pmod{m} $,
\item[(c2)]
$ \bs{f}_{\!\pi} \bs{\tilde{T}}_{\!\pi} \stackrel{?}{=} \bs{b} \pmod{m} $,
\item[(c3)]
$ \wtl(\bs{f}_{\!\pi}) \stackrel{?}{=} w $,
\item[(c4)]
$ \bs{f}_{\!\pi} $ is balanced.
\end{enumerate}
\end{enumerate}
In each of the three options, the Verifier performs the corresponding range checks as
well, e.g., whether $ \bs{f}_{\!\pi} \stackrel{?}{\in} \{-1, 0, 1\}^{n\ell} $.
\item
The Verifier accepts iff all the checks in the selected option are satisfied.
\myqed
\end{enumerate}}

\end{protocol}

Intuitively, Protocol~\ref{protocol} is zero-knowledge because the Verifier obtains
from the Prover only random-looking objects that it could have generated by itself.

\begin{theorem}
Protocol~\ref{protocol} is a zero-knowledge proof for the balanced Lee syndrome
decoding problem.
The soundness error is at most $ 2/3 $.
\end{theorem}
\begin{proof}
We show that all properties of zero-knowledge proofs hold.

\emph{Completeness}.
If the Prover knows the solution $ \bs{e} \in \Z_m^n $ and follows the protocol,
the Verifier will accept the proof as valid because all of its checks will be
satisfied.
This follows from the following facts:
first, $ \bs{f} $ is, by construction, a balanced vector of weight $ \wtl(\bs{f}) = w $
and, hence, so is $ \bs{f}_{\!\pi} $;
second,
\begin{align}
  \bs{a} + \bs{b} = \bs{e} \big( \bs{R} + \bs{T} \big) = \bs{e} \bs{H} = \bs{s} \pmod{m} ;
\end{align}
and third,
\begin{align}
  \bs{f}_{\!\pi} \bs{\tilde{R}}_\pi = \bs{f} \bs{\tilde{R}} = \bs{a} \pmod{m} ,  \qquad
	\bs{f}_{\!\pi} \bs{\tilde{T}}_\pi = \bs{f} \bs{\tilde{T}} = \bs{b} \pmod{m} .
\end{align}
%Note that $ \bs{f}_\pi \bs{\tilde{R}}_\pi = \bs{f} \bs{\tilde{R}} = \bs{a} $ and
%$ \bs{f}_\pi \bs{\tilde{T}}_\pi = \bs{f} \bs{\tilde{T}} = \bs{b} $.

\emph{Soundness}.
We argue that, if the Prover does not possess a vector $ \bs{e} \in \Z_m^n $
satisfying $ \mye \myH = \mys  \pmod{m} $, $ \wtl(\mye) \leqslant w $, and $ \sum_{i=1}^n e_i = 0 $,
then at least one of the checks the Verifier performs will not be satisfied, implying
that if the Verifier selects from among the three options uniformly at random, the
Prover's attempt to cheat will be caught with probability at least $ 1/3 $.
In order to show the contrapositive, suppose that all the conditions (a1--2), (b1--4),
and (c1--4) are satisfied, so that it is impossible to catch the Prover.
This implies that the vector $ \bs{f} $ is balanced (by (b4) and (c4)), is of weight
$ \wtl(\bs{f}) = w $ (by (b3) and (c3)), and satisfies
\begin{align}
  \bs{f} \myHH = \bs{f} \big( \bs{\tilde{R}} + \bs{\tilde{T}} \big) = \bs{a} + \bs{b} = \bs{s} \pmod{m}
\end{align}
(by (a1--2), (b1--2), and (c1--2)).
Then, as shown in the proof of Proposition~\ref{thm:equiv}, its block-wise accumulation
$ \mye = (e_1, \ldots, e_n) $, defined by $ e_i = \sum_{j=(i-1) \ell + 1}^{i\ell} f_j $,
is a solution to the instance $ (\bs{H}, \bs{s}, w) $ of Problem~\ref{problem:balsyn}.
Furthermore, since the Prover knows $ \bs{f} $, it also knows $ \bs{e} $.
Put differently, if the Prover does not possess a solution to the instance
$ (\bs{H}, \bs{s}, w) $ of Problem~\ref{problem:balsyn}, at least one of the conditions
(a1--2), (b1--4), (c1--4) will not be satisfied, meaning that the Verifier will reject
with probability at least $ 1/3 $.

\emph{Zero-knowledge}.
To show that the Verifier learns nothing beyond the fact that the instance
$ (\bs{H}, \bs{s}, w) $ of Problem~\ref{problem:balsyn} has a solution, we argue that
it could have by itself simulated the entire exchange.
To simulate choice (a), the Verifier selects a random matrix $ \myR \in \Z_m^{n\times(n-k)} $,
computes $ \bs{T} = \myH - \myR \pmod{m} $, then selects a random permutation $ \pi $
on $ \{1, \ldots, n\ell\} $ and creates $ \bs{\tilde{R}}_\pi, \bs{\tilde{T}}_{\!\pi} $.
To simulate choice (b), the Verifier selects a random matrix $ \myR \in \Z_m^{n\times(n-k)} $,
creates $ \bs{\tilde{R}} $ and permutes it randomly to obtain $ \bs{\tilde{R}}_{\pi} $.
It then selects uniformly at random a vector $ \bs{g} $ from the set of all balanced
vectors of weight $ w $ from $ \{-1, 0, 1\}^{n\ell} $, and then computes
$ \bs{a} = \bs{g} \bs{\tilde{R}}_{\pi} $ and $ \bs{b} = \bs{s} - \bs{a} \pmod{m} $.
The choice (c) is similar.
It follows that the view of the Verifier during the simulation is identical to its
view during the actual protocol execution.
This, together with the hiding property of $ \com(\cdot) $ (which means that nothing
about $ \bs{R} $, $ \bs{T} $, $ \bs{a} $, $ \bs{b} $, $ \pi $, $\bs{f}_{\!\pi} $ is
revealed from the commitments to these quantities), implies that Protocol~\ref{protocol}
is zero-knowledge.

\emph{Polynomial complexity}.
Finally, notice that the total length of all messages exchanged in the protocol is
polynomial in the size of the problem instance, and that all the checks the Verifier
needs to perform can be done in polynomial time.
\hfill $ \blacksquare $
\end{proof}

As usual, the soundness error can be reduced below the desired threshold by repeating
the protocol sufficiently many times, each time with fresh randomness.
Namely, $ t $ repetitions of the above protocol result in the soundness error not
exceeding $ (2/3)^{t} $.

\subsection{Remarks on Protocol \ref{protocol}}

We next state several remarks on the above protocol and the assumptions that we
adopted along the way.
\begin{enumerate}[itemsep=2mm]
\item
The presented zero-knowledge proof is in fact a zero-knowledge \emph{proof of knowledge}
\cite{goldreich,thaler}.
Namely, using Protocol~\ref{protocol}, the Prover convinces the Verifier not only
that a solution to the given instance of Problem~\ref{problem:balsyn} exists, but
that the Prover \emph{knows} it.
This fact can be seen from the proof of the soundness property where we implicitly
built an ``extractor'' that, given access to the prover algorithm with a rewind option,
finds a witness with high probability.
In other words, we have in fact proved the stronger ``knowledge soundness'' property.

\item
The presented zero-knowledge proof can be applied to the general syndrome
decoding problem (Problem~\ref{problem:syn}) by using the reduction from
Proposition~\ref{thm:NPcomplete}.
Namely, given an instance $ (\bs{H}, \bs{s}, w) $ of Problem~\ref{problem:syn},
the Prover and the Verifier would create an instance
$ (\bs{H^{\sml{\pm}}}, 2\bs{\bar{s}}, 2w) $ of Problem~\ref{problem:balsyn}.
Further, possessing a witness $ \bs{e} $, the Prover would form a balanced vector
$ (\bs{\bar{e}} \,| -\!\bs{\bar{e}}) $ (a witness for the just created instance
of Problem~\ref{problem:balsyn}), and then proceed with the protocol as described
above.
We have chosen to state the protocol for Problem~\ref{problem:balsyn} in order
to simplify the exposition, as well as to emphasize the need for the additional
assumptions adopted in Problem~\ref{problem:balsyn} for the execution of the
protocol.

\item
Why is the transformation of $ \bs{e} \in \Z_m^n $ to
$ \bs{f} = \bs{e''} \in \{-1, 0, 1\}^{n\ell} $ performed by the Prover?
For example, why doesn't the Prover select a random permutation $ \mu $ on
$ \{1, \ldots, n\} $ and commit to $ \bs{R}_\mu, \bs{T}_{\!\mu}, \bs{e}_\mu $
(instead of $ \bs{\tilde{R}}_\pi, \bs{\tilde{T}}_{\!\pi}, \bs{f}_{\!\pi} $)?
It is easy to see that such a protocol would not possess the zero-knowledge
property because sending $ \bs{e}_\mu $ to the Verifier would reveal the
composition of the vector $ \bs{e} $.

\item
Why are the assumptions that $ \bs{e} $ is balanced and of weight upper-bounded
by $ n (\ell - 1) $ needed?
Under these two assumptions, one is guaranteed to be able to transform $ \bs{e} \in \Z_m^n $
to a vector $ \bs{f} = \bs{e''} \in \{-1, 0, 1\}^{n\ell} $ that is also balanced
and of weight exactly $ w $ (see the proof of Proposition~\ref{thm:equiv}).
A random permutation of such a vector, $ \bs{f}_{\!\pi} $, which the Verifier may
require to see, is a random vector containing $ w/2 $ ones, $ w/2 $ negative ones,
and $ n\ell - w $ zeros, regardless of $ \bs{e} $.
As such, it reveals absolutely nothing about $ \bs{e} $ that the Verifier didn't
already know.

\item
Why are the Verifier's queries (b) and (c) not combined into a single query where
the Prover would be asked to open $ \com(\mya) $, $ \com(\myb) $,
$ \com\!\big(\bs{\tilde{R}}_\pi\big) $, $ \com\!\big(\bs{\tilde{T}}_{\!\pi}\big) $,
and $ \com(\bs{f}_{\!\pi}) $?
The reason is that $ \bs{\tilde{R}} $ and $ \bs{\tilde{T}} $ are not independent
(because their sum is the matrix $ \bs{\tilde{H}} $ known to the Verifier), and
hence the Verifier might, after seeing $ \bs{\tilde{R}}_\pi $ and $ \bs{\tilde{T}}_{\!\pi} $,
learn something about the permutation $ \pi $, and then from $ \bs{f}_{\!\pi} $
learn something about $ \bs{f} = \bs{e''} $, and therefore about $ \bs{e} $.
In other words, such a protocol would not be zero-knowledge.
Likewise, the queries (a) and (b) (or (a) and (c)) cannot be combined as the
corresponding commitment openings would trivially reveal $ \bs{f} $, and therefore
$ \bs{e} $, to the Verifier.

\item
The soundness error of Protocol \ref{protocol} cannot be reduced by simply removing
one of the queries (a),(b),(c) of the Verifier, or by combining two of them into one.
As explained in the previous remark, combining any two queries would result in a protocol
that is not zero-knowledge.
It is also easy to see that, after removing any of the queries, the protocol would
no longer be sound, i.e., the Prover would be able to cheat the Verifier with probability
$ 1 $.

\item
In the case of even $ m = 2 \ell $, the elements of $ \Z_m $ are again represented
by the integers $ \{-\ell, \ldots, \ell\} $ and the analysis proceeds in the same
way.
A minor issue is that now $ -\ell $ and $ \ell $ represent the same element of $ \Z_m $
and either of these two representations may be chosen freely in any situation.
This, however, does not cause any inconsistencies in the protocol.

The only place where a modification is needed for even $ m $ is the reduction
described in the proof of Proposition~\ref{thm:NPcomplete}.
Namely, in this case, $ 2 $ is not invertible in $ \Z_m $.
This can be circumvented by defining
\begin{align}
\label{eq:Hpm2}
  \bs{H^{\sml{\pm}}} =
  \begin{pmatrix}
     \bs{\bar{H}}  \\
  \hline
     -(c-1)\bs{\bar{H}}\,
  \end{pmatrix} ,
\end{align}
where $ c $ is any number satisfying $ \gcd(m,c) = 1 $.
Now the same steps as in the proof of Proposition~\ref{thm:NPcomplete} establish
that $ (\myH, \mys, w) $ is a `yes' instance of the Lee syndrome decoding problem
(Problem~\ref{problem:syn}) if and only if $ (\bs{H^{\sml{\pm}}}, c\bs{\bar{s}}, 2w) $
is a `yes' instance of the balanced Lee syndrome decoding problem (Problem~\ref{problem:balsyn}).

\item
The communication complexity of Protocol~\ref{protocol} is determined by steps 5 and 6.
Namely, the Prover needs to send $ 8 $ commitments and up to $ 5 $ openings.
Assuming that the commitments are of constant size (e.g., hashes), the communication
complexity is dominated by the openings $ \bs{R} $, $ \bs{T} $, $ \pi $, $ \bs{\tilde{R}}_\pi $,
$ \bs{\tilde{T}}_{\!\pi} $, whose total size in bits is
$ 2 n (n-k) \log_2(m) + n\ell \log_2(n\ell) + 2 n \ell (n-k) \log_2(m) $.
An example with concrete numbers is hard to give at this point because the values of
the parameters needed for a desired security level for the Lee syndrome decoding problem
do not seem to have been thoroughly explored yet.
In \cite[Example 19]{horlemann2}, codes with $ n = 425 $, $ k = 229 $ over $ \Z_4 $
(i.e., with $ m = 4 $ and $ \ell = \lfloor m/2 \rfloor $ = 2) are claimed to exist
for which the security level of $ 128 $ bits is reached.
Using these values, the above expression for the communication complexity equals
approximately $ 1 $ Mb.

\end{enumerate}

\section{Conclusion}
\label{sec:conclusion}

Motivated by the relevance of both code-based (and, more generally, post-quantum)
cryptography and zero-knowledge proofs, we have described a zero-knowledge proof
of knowledge for the Lee-metric version of the syndrome decoding problem.
Our intention was to present a conceptually simple protocol, rather than to propose
a performance-competitive solution, which is why we have not attempted to optimize
the protocol for efficiency.
In particular, as commented in Remark 8 above, the communication costs are quite
high, and the protocol should therefore be regarded as a preliminary solution.
Its potential use in practical cryptographic constructions would require detailed
optimization of the whole procedure and the various involved parameters, as well as
conversion to a non-interactive form (e.g., via the Fiat--Shamir transform \cite{fiat+shamir}).
Further efficiency improvements could be achieved by using standard techniques such
as weighted challenges (used to reduce the number of repetitions needed for soundness
amplification) and seed trees (which allow for compact representation of randomness
in multi-round protocols), as in, e.g., \cite{baldi_NIST}.

Our further work on the subject will focus on developing more efficient non-interactive
versions of the protocol, suitable for modern applications.
This will also include exploring alternative approaches, such as the shared permutation
framework \cite{feneuil,bidoux2}, or the MPC-in-the-Head paradigm and its efficient
instantiations such as VOLE-in-the-Head \cite{baum} or TC-in-the-Head \cite{feneuil2}.

\vspace{5mm}
%\begin{credits}
{\small
\subsubsection{\small Acknowledgments.}
%This paper is dedicated to students and teachers who bravely stood against corruption
%and injustice in Serbia during the academic year 2024/25.
%who rose bravely against oppression, corruption, and injustice in Serbia.
This research was supported by the Ministry of Science, Technological Development
and Innovation of the Republic of Serbia (contract no. 451-03-137/2025-03/200156)
and by the Faculty of Technical Sciences, University of Novi Sad, Serbia (project
no. 01-50/295).
%through project
%``Scientific and Artistic Research Work of Researchers in Teaching and
%Associate Positions at the Faculty of Technical Sciences, University of
%Novi Sad 2025'' (No. 01-50/295).

%\subsubsection{\small Disclosure of Interests.}
%The authors have no competing interests to declare that are relevant to the content
%of this article.
%It is now necessary to declare any competing interests or to specifically
%state that the authors have no competing interests. Please place the
%statement with a bold run-in heading in small font size beneath the
%(optional) acknowledgments\footnote{If EquinOCS, our proceedings submission
%system, is used, then the disclaimer can be provided directly in the system.},
%for example: 
%The authors state that they do not have any conflicts of interest concerning the publication of this article.
%Or: Author A has received research grants from Company W.
%Author B has received a speaker honorarium from Company X and owns stock in Company Y.
%Author C is a member of committee Z.
%\end{credits}
}


\end{document}